\documentclass{siamart1116}

\usepackage{lipsum}
\usepackage{amsfonts}
\usepackage{graphicx}
\usepackage{epstopdf}
\usepackage{algorithmic}
\ifpdf
  \DeclareGraphicsExtensions{.eps,.pdf,.png,.jpg}
\else
  \DeclareGraphicsExtensions{.eps}
\fi
\usepackage{amsmath} 
\usepackage{amssymb}  
\usepackage{subcaption}

\numberwithin{theorem}{section}

\newcommand{\TheTitle}{Optimal containment of epidemics over temporal activity-driven networks} 
\newcommand{\TheAuthors}{M. Ogura, V. M. Preciado, and N. Masuda}

\headers{Containing epidemics on activity-driven networks}{\TheAuthors}

\title{{\TheTitle}\thanks{Submitted to the editors DATE.
}}

\author{
  Masaki Ogura\thanks{Graduate School of Information Science, Nara Institute of Science and Technology, Ikoma, Nara (\email{oguram@is.naist.jp}).}
\and 
Victor M. Preciado\thanks{Department of Electrical and Systems Engineering, University of Pennsylvania, Philadelphia, PA USA (\email{preciado@seas.upenn.edu}).}
  \and
  Naoki Masuda\thanks{Department of Engineering Mathematics, University of Bristol, Clifton, Bristol BS8 1UB (\email{naoki.masuda@bristol.ac.uk}).}
}

\usepackage{amsopn}
\usepackage{graphicx}
\usepackage{color}
\usepackage{subcaption}
\usepackage{mathtools}
\mathtoolsset{centercolon}
\usepackage[noadjust]{cite}
\usepackage{enumitem}

\usepackage{siunitx}

\DeclareSymbolFont{bbold}{U}{bbold}{m}{n}
\DeclareSymbolFontAlphabet{\mathbbold}{bbold}
\newcommand{\onev}{\mathbbold{1}}

\DeclareMathOperator{\st}{subject\ to}
\DeclareMathOperator*{\minimize}{minimize}

\newcommand{\norm}[1]{\lVert{#1}\rVert}
\newcommand{\abs}[1]{\lvert{#1}\rvert}
\newcommand{\av}[1]{\langle{#1}\rangle}

\newsiamthm{claim}{Claim} 
\newsiamremark{remark}{Remark} 
\newsiamremark{problem}{Problem} 
\newsiamremark{conjecture}{Conjecture} 
\newsiamremark{assumption}{Assumption} 
\crefname{hypothesis}{Hypothesis}{Hypotheses} 
\crefname{problem}{Problem}{Problems} 

\usepackage{accents}
\newcommand{\ubar}[1]{\underaccent{\bar}{#1}} 
\newcommand{\myubar}[2]{\raisebox{-#2}{\underbar{\raisebox{#2}{#1}}}}

\ifpdf
\hypersetup{
  pdftitle={\TheTitle},
  pdfauthor={\TheAuthors}
}
\fi

\begin{document}

\maketitle

\begin{abstract}
In this paper, we study the dynamics of epidemic processes taking place in temporal and adaptive networks. Building on the activity-driven network model, we propose an adaptive model of epidemic processes, where the network topology dynamically changes due to both exogenous factors independent of the epidemic dynamics as well as endogenous preventive measures adopted by individuals in response to the state of the infection. A direct analysis of the model using Markov processes involves the spectral analysis of a transition probability matrix whose size grows exponentially with the number of nodes. To overcome this limitation, we derive an upper-bound on the decay rate of the number of infected nodes in terms of the eigenvalues of a $2 \times 2$ matrix. Using this upper bound, we propose an efficient algorithm to tune the parameters describing the endogenous preventive measures in order to contain epidemics over time. We confirm our theoretical results via numerical simulations.
\end{abstract}




\newcommand{\af}{\chi}
\newcommand{\ap}{\pi}
\newcommand{\costAF}{f}
\newcommand{\costAP}{g}

\section[Introduction]{Introduction}

Accurate prediction and cost-effective containment of epidemics in human and animal populations are fundamental problems in mathematical epidemiology~\cite{Pastor-Satorras2015a,Diekmann2000,Kiss2017}. In order to achieve these goals, it is indispensable to develop effective mathematical models describing the spread of disease in human and animal contact networks~\cite{Funk2010,Gross2008}. 
In this direction, we find a broad literature on modeling, analysis, and containment of epidemic processes in static contact networks. However, these works neglect an important factor: the temporality of the interactions~\cite{Isella2011a,Stehle2011,Sun2013a}, which arises either independently of or dependent on epidemic propagations. A framework for modeling temporal interactions in human and animal populations is temporal networks (i.e., time-varying networks), where individuals and interactions are modeled as nodes and edges, respectively, which can appear and disappear over time~\cite{Holme2015b,Masuda2016b,Holme2012}. Under this framework, the effect of temporal interactions on epidemic propagations has been investigated numerically and theoretically~\cite{Masuda2017,Masuda2013}. For containing epidemic processes on temporal networks, we find a plethora of heuristic approaches~\cite{Prakash2010,Lee2012} and analytical methods~\cite{Ogura2015c,Liu2014a}.

Adaptive networks refer to the case in which changes in nodes or edges occur in response to the state of the dynamics taking place on the network~\cite{Masuda2016b,Gross2008,Gross2009,Sayama2013}. A common temporality of agent-agent interaction in epidemic dynamics arises from social distancing behavior~\cite{Aledort2007,Bootsma2007,Bell2006}, which let the structure of contact networks change over time as a result of adaptation to the state of the epidemics. Several models of such adaptive networks have been proposed. For example, Gross et al.~proposed a rewiring mechanism where a healthy node actively avoids to be adjacent to infected nodes~\cite{Gross2006}. Extensions of this model are found in~\cite{Gross2008,Zanette2008a,Marceau2010,Lagorio2011,Tunc2014}. Guo et al.~proposed an alternative model in which links connecting an infected node and a healthy node are deactivated~\cite{Guo2013}. As for the containment of epidemic processes on adaptive networks, various heuristic~\cite{Bu2013,Maharaj2012} and  analytical~\cite{Ogura2015i,Ogura2016l} approaches have been proposed. However, in many studies, the effects of exogenous temporal factors and endogenous adaptive measures on epidemic processes have been separately examined, leaving unclear how their combination affects the dynamics of the spread.

In this paper, we study epidemic processes and containment strategies in a temporal network model where the effect of exogenous factors and that of adaptive measures are simultaneously present. Our model is based on the activity-driven temporal network model~\cite{Perra2012}. In this model, a node is stochastically activated and connects to other nodes independently of the dynamics taking place in the network. In order to analyze the joint effect of exogenous factors and endogenous adaptations, we add a mechanism of social distancing to the standard activity-driven model. In other words, we allow an infected node to endogenously adapt to the state of the epidemics by  1)~decreasing its activation probability and 2) refusing interactions with other activated nodes. On top of this temporal network, we adopt the standard susceptible-infected-susceptible (SIS) model of epidemic dynamics~(see, e.g., \cite{Pastor-Satorras2015a}) and derive an analytical upper bound on the decay rate of the number of infected nodes over time. Based on this result, we then propose an efficient strategy for tuning the social distancing rates in order to suppress the number of infected nodes.

Our work is related to \cite{Rizzo2014}, in which an infected individual is allowed to decrease its activation probability. However, in \cite{Rizzo2014}, the durations of temporal interactions are assumed to be sufficiently short compared with the time scale of the epidemic dynamics, leaving out the interesting case where the time scale of the network dynamics and that of the epidemic process are comparable. In addition, our results hold true for networks of any size, while the results in~\cite{Rizzo2014} require the networks to be sufficiently large.

This paper is organized as follows. In \Cref{sec:prbSetting}, we introduce a model of epidemic processes on temporal and adaptive networks. In \Cref{sec:decayRate}, we derive an upper bound on the decay rate of the infection size. Based on this bound, in \Cref{sec:optimizaiton} we formulate and solve optimization problems for containing the spread of epidemics. The obtained theoretical results are numerically illustrated in \Cref{ref:numerical}.

\section{Problem setting}\label{sec:prbSetting}

In this section, we first describe the activity-driven network proposed in~\cite{Perra2012}. We then introduce an adaptive SIS (A-SIS) model on activity-driven networks, which allows nodes to adapt to the state of the nodes (i.e., susceptible or infected) in their neighborhoods.

\subsection{Activity-driven networks}\label{sec:adsis}

Throughout this paper, we let the set of nodes in a network be given by $\mathcal V = \{v_1, \dotsc, v_n \}$. The activity-driven network is a temporal network in discrete time and is defined as follows.

\begin{figure}[tb]
\centering \includegraphics[clip,trim={2cm 4.2cm 1.2cm
6.8cm},width=.65\linewidth]{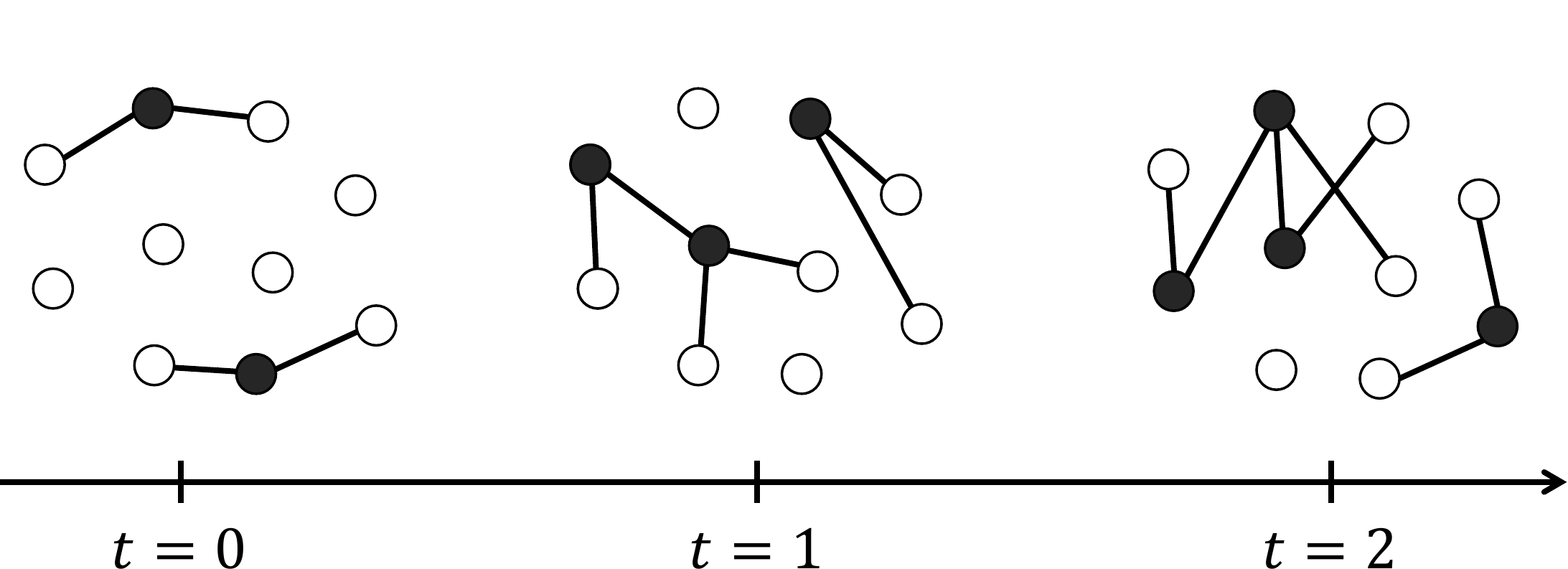} \caption{Schematic on an activity-driven
network. We set $n=10$ and $m=2$. Filled circles represent active nodes. Open circles represent inactive nodes. The time is denoted by $t$.}
\label{fig:adn}
\end{figure}

\begin{definition}[\cite{Perra2012}]\label{eq:ADM}
For each $i=1, \dotsc, n$, let $a_i$ be a positive constant less than or equal to $1$. We call $a_i$ the \emph{activity rate} of node $v_i$. Let $m$ be a positive integer less than or equal to $n-1$. The \emph{activity-driven network} is defined as an independent and identically distributed sequence of undirected graphs created by the following procedure (see \cref{fig:adn} for a schematic illustration):
\begin{enumerate}
\item At each time $t = 0, 1, 2, \dotsc$, each node $v_i$ becomes ``activated'' with
probability~$a_i$ independently of other nodes.

\item 
Each activated node, say, $v_i$, randomly and uniformly chooses $m$ other nodes independently of other activated nodes. For each chosen node, say, $v_j$, an edge $\{v_i, v_j\}$ is created. These edges are discarded at time $t+1$ (i.e., do not exist at time $t+1$).

\item Steps 2 and 3 are repeated for each time $t \geq 0$, independently of
past realizations.
\end{enumerate}
\end{definition}

\begin{remark}\label{rmk:}
We do not allow multiple edges between a pair of nodes. In other
words, even when a pair of activated nodes choose each other as their neighbors
at a specific time, we assume that one and only one edge is spanned between
those nodes. 
\end{remark}

Although the activity-driven network is relatively simple, the model can reproduce an arbitrary degree distribution~\cite{Perra2012}. Several properties of activity-driven networks have been investigated, including structural properties~\cite{Perra2012,Starnini2013b}, steady-state properties of random walks~\cite{Perra2012a,Ribeiro2013}, and spreading dynamics~\cite{Perra2012,Rizzo2014,Speidel2016a}. However, the model does not allow nodes to adapt to the state of the epidemics and, therefore, is not suitable for discussing how social distancing affects the dynamics of the spread. In the next subsection, we extend the activity-driven network by incorporating social distancing behaviors of nodes.

\subsection{Activity-driven A-SIS model}

Building upon the activity-driven network described above, we consider the scenario where nodes change their neighborhoods in response to the state of the epidemics over the network~\cite{Ogura2015i}. Specifically, we propose the \emph{activity-driven adaptive-SIS model} (\emph{activity-driven A-SIS model} for short) as follows:

\begin{figure}[tb]
\centering 
\includegraphics[clip,trim={5.5cm 6.2cm 6.5cm 3.7cm},width=.55\linewidth]{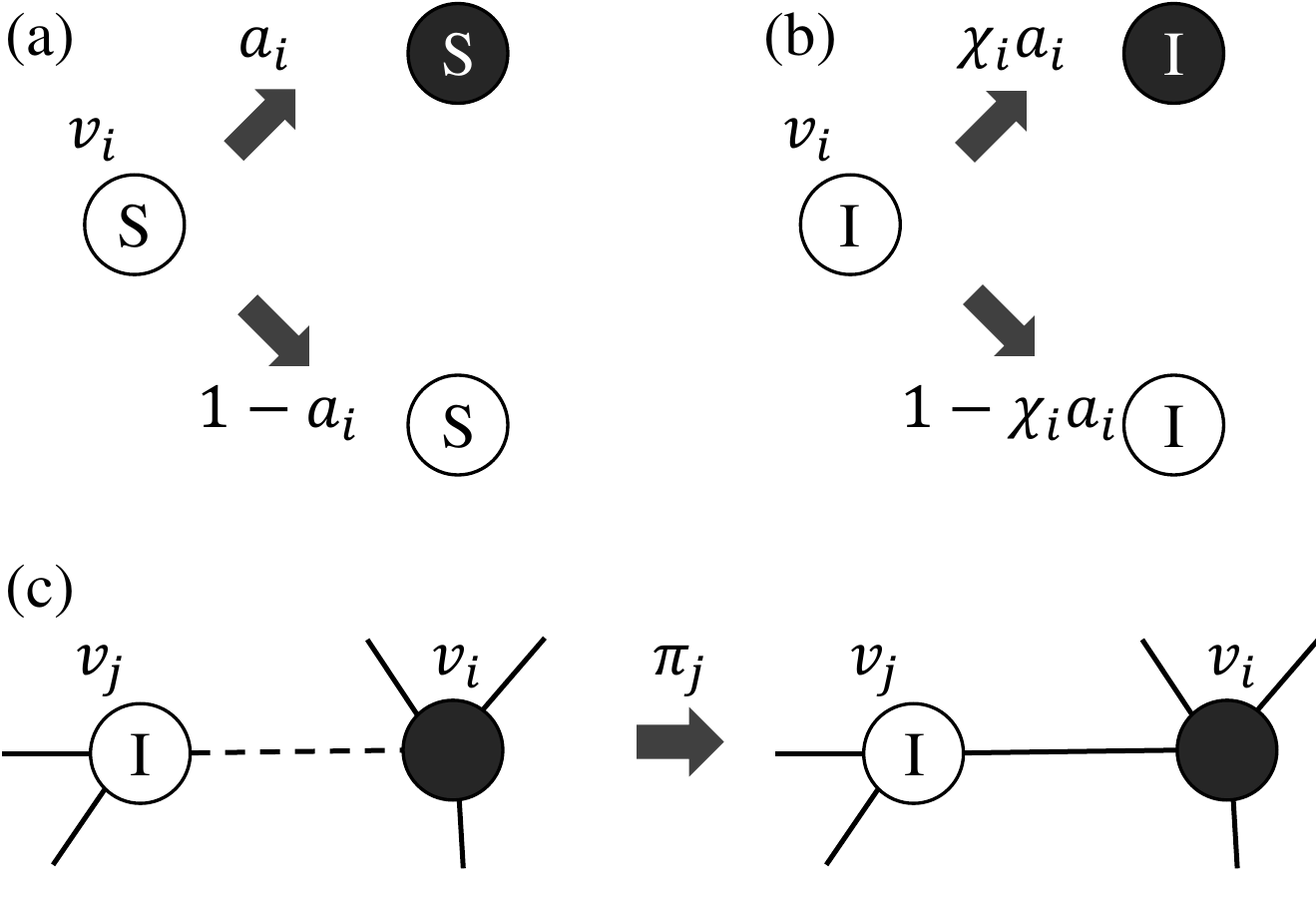} 
\caption{Adaptation of nodes in the activity-driven A-SIS model. Filled and empty circles represent active and inactive nodes, respectively. (a) A susceptible node is activated with probability~$a_i$. (b) An infected node is activated with probability~$\af_i a_i$. (c) An infected node ($v_j$) accepts an edge spanned from an activated node with probability~$\ap_j$.}\label{fig:adaptiveadn}
\end{figure}

\begin{definition}[Activity-driven A-SIS model]\label{defn:adasis}
For each $i$, let $a_i, \af_i, \ap_i \in (0, 1]$ be constants. We call $a_i$, $\af_i$, and $\ap_i$ the \emph{activity rate}, \emph{adaptation factor}, and \emph{acceptance rate} of node $v_i$, respectively. Also, let $m \leq n-1$ be a positive integer, and $\beta, \delta \in (0, 1]$ be constants. We call $\beta$ and $\delta$ the \emph{infection rate} and \emph{recovery rate}, respectively. The \emph{activity-driven A-SIS model} is defined by the following procedures (see \cref{fig:adaptiveadn} for an illustration):

\begin{enumerate}
\item At the initial time $t=0$, each node is either \emph{susceptible} or \emph{infected}.

\item\label{item:lessActivation} At each time $t = 0, 1, 2, \dotsc$, each node $v_i$ randomly becomes activated independently of other nodes with the following probability:
\begin{equation}
\Pr(\mbox{node $v_i$ becomes activated}) = 
\begin{cases}
a_i,& \mbox{if $v_i$ is susceptible,}
\\
\af_ia_i,& \mbox{if $v_i$ is infected.}
\end{cases}
\end{equation} 

\item\label{item:cutting} Each activated node, say, $v_i$, randomly and uniformly chooses $m$ other nodes independently of other activated nodes. For each chosen node, say, $v_j$, an edge $\{v_i, v_j\}$ is created with the following probability:
\begin{equation}\label{eq:acceptanceProb}
\Pr(\mbox{$\{v_i, v_j\}$ is created}) = 
\begin{cases}
1, & \mbox{if $v_j$ is susceptible,}
\\
\ap_j, & \mbox{if $v_j$ is infected. }
\end{cases}
\end{equation}
These edges are discarded at time $t+1$ (i.e., do not exist at time $t+1$). As
in \cref{rmk:}, we do not allow multiple edges between a pair of
nodes.

\item The states of nodes are updated according to the SIS model. In other words, if a node~$v_i$ is infected, it transits to the susceptible state with probability~$\delta$. If $v_i$ is susceptible, its infected neighbors infect node $v_i$ with probability~$\beta$ independently of the other infected neighbors. 

\item Steps 2--4 are repeated for each time $t \geq 0$. 
\end{enumerate}
\end{definition}

Steps~\labelcref{item:cutting,item:lessActivation} in \cref{defn:adasis} model social distancing behavior by infected nodes. In Step~\labelcref{item:lessActivation}, an infected node decreases its activity rate to avoid infecting other nodes. Step~\labelcref{item:lessActivation} can also be regarded as modeling reduction of social activity by infected nodes due to sickness. In Step~\labelcref{item:cutting}, an infected node, say, $v_j$, establishes a connection with an activated node only with probability~$\ap_j$ to avoid infecting other nodes (when $\pi_j < 1$). A susceptible node behaves in the same way as in the standard SIS model in the original activity-driven network.

\section{Decay rate}\label{sec:decayRate}

In order to quantify the persistence of epidemic infections in the activity-driven A-SIS model, in this section, we introduce the concept of decay rate of the epidemics. A direct computation of the decay rate requires computing the eigenvalues of a matrix whose size grows exponentially with the number of the nodes. To overcome this difficulty, we present an upper bound on the decay rate in terms of the eigenvalues of a $2\times 2$ matrix.

\subsection{Definition}

For each time $t$ and node $v_i$, define the random variable
\begin{equation}
x_i(t) = \begin{cases}
0, & \mbox{if $v_i$ is susceptible at time $t$,}
\\
1, & \mbox{if $v_i$ is infected at time $t$}. 
\end{cases}
\end{equation} 
Define the vector $p(t) = [p_1(t)\ \cdots \ p_n(t)]^\top$ of the infection probabilities by
\begin{equation}\label{eq:def:p_i}
p_i(t) = \Pr(\mbox{$v_i$ is infected at time $t$}). 
\end{equation}
In this paper, we measure the persistence of infection by the rate of convergence of infection probabilities to the origin. 

\begin{definition}
We define the \emph{decay
rate} of the activity-driven A-SIS model by 
\begin{equation}
\alpha  = \sup_{x(0)}\limsup_{t\geq 0} \frac{\log \norm{p(t)}}{t}, 
\end{equation}
where $\norm{\cdot}$ denotes the $\ell_1$ norm.
\end{definition}

The infection-free equilibrium, $x_1 = \cdots = x_n = 0$, is the unique absorbing state of the Markov process $\{x_1(t), \dotsc, x_n(t)\}_{t\geq 0}$. Moreover, the infection-free equilibrium is reachable from any other states by our assumption $\delta > 0$. This implies $\alpha < 1$. In fact, $\alpha$ is difficult to compute for large networks for the following reason. The Markov process~$\{x_1(t), \dotsc, x_n(t)\}_{t\geq 0}$ has $2^n$ states. Let $Q$ denote its $2^n \times 2^n$ transition probability matrix. Since the disease-free state is the unique absorbing state, it follows that
\begin{equation}
\alpha = \max \{ \abs{\lambda}: \mbox{$\lambda$ is an eigenvalue of $Q$,\ $\abs{\lambda}<1$}   \}. 
\end{equation}
Because the size of the matrix~$Q$ grows exponentially fast with respect to the number of the nodes, a direct computation of the decay rate is difficult for large networks.

\subsection{An upper bound}

We start with the following proposition, which allows us to upper-bound the infection probabilities using a linear dynamics:

\begin{proposition}\label{prop:dynamics}
Let 
\begin{equation}\label{eq:def:barm}
\bar m = m/(n-1)
\end{equation}
and, for all $i$, define the constants
\begin{equation}\label{eq:def:phipsi}
\phi_i = \bar m \af_i a_i,\quad \psi_i = \bar m \ap_i a_i. 
\end{equation}
Then, 
\begin{equation}\label{eq:upperDynamics}
p_i(t+1) 
\leq 
(1-\delta) p_i(t) + \beta \sum_{j=1}^n[1 -
(1-\psi_i)(1-\phi_j)] p_j(t)
\end{equation}
for all nodes $v_i$ and $t\geq 0$.
\end{proposition}

\begin{proof} 
By the definition of the A-SIS dynamics on the activity-driven network, the nodal states $x_1$, \dots, $x_n$ obey the following stochastic
difference equation
\begin{equation}\label{eq:originalDynamics}
x_i(t+1) = x_i(t) - x_i(t) N_{\delta}^{(i)}(t) + (1-x_i(t)) 
\left[1-\prod_{j\neq i} \left(1-a_{ij}(t) x_j(t) N_{\beta}^{(ij)}(t)\right)\right], 
\end{equation}
where 
\begin{equation}
a_{ij}(t) = \begin{cases}
1,& \mbox{if an edge $\{v_i, v_j\}$ exists at time $t$,}
\\
0, & \mbox{otherwise,}
\end{cases}
\end{equation}
and $\{N_{\delta}^{(i)}(t)\}_{t = 0}^\infty$ and $\{N_{\beta}^{(ij)}(t)\}_{t = 0}^\infty$ are independent and identically distributed random Bernoulli variables satisfying
\begin{equation}
N_{\delta}^{(i)}(t) = 
\begin{cases}
1, & \mbox{with probability~$\delta$}, 
\\
0, & \mbox{with probability~$1-\delta$}, 
\end{cases}
\end{equation}
and 
\begin{equation}
N_{\beta}^{(ij)}(t) = 
\begin{cases}
1, & \mbox{with probability~$\beta$}, 
\\
0, & \mbox{with probability~$1-\beta$}. 
\end{cases}
\end{equation}
On the right-hand side of equation~\cref{eq:originalDynamics}, the second and third terms represent recovery and transmission events, respectively (a similar equation for the case of static networks can be found in~\cite{Chakrabarti2008}).

By the Weierstrass product inequality, the third term on the right-hand side of \cref{eq:originalDynamics} is upper-bounded by $(1-x_i(t))\sum_{j=1}^n a_{ij}(t) x_{j}(t) N_{\beta}^{(ij)}(t)$. Since the expectation of $x_i(t)$ equals $p_i(t)$, taking the expectation in \cref{eq:originalDynamics} gives
\begin{equation}\label{eq:p_i(t+1)<=...}
p_i(t+1) 
\leq
p_i(t) - \delta p_i(t)  + \beta 
\sum_{j\neq i} E[(1-x_i(t))a_{ij}(t) x_j(t)], 
\end{equation}
where $E[\cdot]$ denotes the expectation of a random variable. 

Now, assume $i\neq j$. By the definition of the variables $x_i$ and $a_{ij}$, it follows  that
\begin{equation}\label{eq:E[]=...}
\begin{aligned}
&E[(1-x_i(t)) a_{ij}(t) x_j(t)]
\\
=&
\Pr(\mbox{$v_i$ and $v_j$ are adjacent, $v_i$ is susceptible, and $v_j$ is infected at time $t$})
\\
=&
\Pr(\mbox{$v_i$ and $v_j$ are adjacent at time $t$} \mid \Xi^t_{i, j})
\Pr(\Xi^t_{i, j}), 
\end{aligned}
\end{equation}
where the event $\Xi^t_{i, j}$ is defined by 
\begin{equation}
\Xi^t_{i, j} = \mbox{``$v_i$ is susceptible and $v_j$ is infected at time $t$''}.
\end{equation}
If we further define the event
\begin{equation}
\Gamma_{i\to j} ^t
= 
\mbox{``$v_i$ is activated and chooses $v_j$ as its neighbor at time $t$''}, 
\end{equation}
then, we obtain 
\begin{equation}\label{eq:adjProbability}
\begin{aligned}
&\Pr(\mbox{$v_i$ and $v_j$ are adjacent at time $t$} \mid \Xi^t_{i, j})
\\
=&
\Pr(\Gamma_{i\to j}^t  \mid \Xi^t_{i, j})
+
\Pr(\Gamma^t_{ j\to i}  \mid \Xi^t_{i, j})
-
\Pr(\Gamma_{i\to j}^t  \mid \Xi^t_{i, j})
\Pr(\Gamma^t_{ j\to i}  \mid \Xi^t_{i, j})
\\
=&
1 - \bigl[1-\Pr(\Gamma_{i\to j}^t  \mid \Xi^t_{i, j})\bigr]
\bigl[1 - \Pr(\Gamma^t_{ j\to i}  \mid \Xi^t_{i, j})\bigr]. 
\end{aligned}
\end{equation}
The event $\Gamma_{i\to j}^t$ occurs when and only when $v_i$ is activated, chooses $v_j$ as a potential neighbor, and actually connects to~$v_j$ (according to the probability given by equation~\cref{eq:acceptanceProb}). Therefore, equation~\cref{eq:def:phipsi} implies
\begin{equation}\label{eq:psi_i}
\Pr(\Gamma_{i\to j}^t \arrowvert
\Xi^t_{i, j}) = \psi_i. 
\end{equation} 
Similarly, the event $\Gamma^t_{ j\to i}$ occurs when and only when $v_j$ is activated (with probability~$\af_j a_j$ if $v_j$ is infected at time $t$) and chooses $v_i$ as one of its $m$ neighbors. Therefore, we have
\begin{equation}\label{eq:phi_j}
\Pr(\Gamma^t_{ j\to i} \arrowvert \Xi^t_{i, j})
= \phi_j.
\end{equation}
Hence, for $i\neq j$, combination of equations \cref{eq:E[]=...,eq:adjProbability,eq:psi_i,eq:phi_j} yields
\begin{equation}\label{eq:pre:upperDynamics}
\begin{aligned}
E[(1-x_i(t)) a_{ij}(t) x_j(t)] &= [1-(1-\psi_i)(1-\phi_j)]\Pr(\Xi_{i,j}^t)
\\
&\leq [1-(1-\psi_i)(1-\phi_j)]p_j(t), 
\end{aligned}
\end{equation}
where we have used the trivial inequality~$\Pr(\Xi_{i,j}^t) \leq p_j(t)$. Moreover, inequality \cref{eq:pre:upperDynamics} trivially holds true also when $i= j$. Inequalities \cref{eq:pre:upperDynamics,eq:p_i(t+1)<=...} prove \cref{eq:upperDynamics}, as desired.
\end{proof}

Using Proposition~\ref{prop:dynamics}, we obtain the following theorem that gives an explicit upper bound on the decay rate of the activity-driven A-SIS model. For a vector~$\xi \in \mathbb R^n$, introduce the notations
\begin{equation}
\av{\xi}_{\!a} = \frac{1}{n}\sum_{i=1}^n a_i\xi_i
,\quad 
\av{\xi}_{\!a^2} = \frac{1}{n}\sum_{i=1}^n a_i^2\xi_i. 
\end{equation}

\begin{theorem}\label{thm:eqanalysis}
Define
\begin{equation}\label{eq:upperBound}
{\alpha_{\rm u}} =  1 - \delta + \kappa \bar m n \beta, 
\end{equation}
where 
\begin{equation}\label{defn:anglers}
\kappa = 
\frac{\av{\chi}_{\!a} + \av{\pi}_{\!a} - \bar m\av{\chi\pi}_{\!a^2} + \sqrt{
{(\av{\chi}_{\!a} + \av{\pi}_{\!a} - \bar m\av{\chi\pi}_{\!a^2})}^2 + 4(\av{\chi\pi}_{\!a^2} - \av{\chi}_{\!a} \av{\pi}_{\!a}) 
}}{2}. 
\end{equation}
Then, the decay rate $\alpha$ satisfies
\begin{equation}
\alpha \leq \alpha_{\rm u}. 
\end{equation}
\end{theorem}

\begin{proof}
Inequality \cref{eq:upperDynamics} implies that there exists a nonnegative variable $\epsilon_i(t)$ such that
\begin{equation}\label{eq:pDynamics+epsilon}
p_i(t+1) 
=
(1-\delta) p_i(t) + \beta \sum_{j=1}^n\left(1 -
(1-\psi_i)(1-\phi_j)\right) p_j(t) - \epsilon_i(t)
\end{equation}
for all nodes $v_i$ and $t\geq 0$. Let us define the vectors $\epsilon(t) = [\epsilon(t)\ \cdots\ \epsilon_n(t)]^\top$, $\phi = [\phi_1\ \cdots \ \phi_n]^\top$, and $\psi = [\psi_1\ \cdots \ \psi_n]^\top$. Equation~\cref{eq:pDynamics+epsilon} is rewritten as 
\begin{equation}\label{eq:originalDynamicsp}
p(t+1) 
=
\mathcal F p(t)- \epsilon(t),
\end{equation}
where 
\begin{equation}\label{eq:def:calF}
\mathcal F =  (1-\delta) I + \beta \left[\onev\onev^\top - (\onev - \psi)( \onev - \phi)^\top\right], 
\end{equation}
$\onev$ denotes the $n$-dimensional vector whose entries are all one, and $I$ denotes the $n\times n$ identity matrix. Since $\mathcal F$ and $\epsilon(t)$ are nonnegative entrywise, equation~\cref{eq:originalDynamicsp} leads to $p(t) = \mathcal F^t p(0) - \sum_{\ell=0}^t \mathcal F^{k-\ell} \epsilon(\ell) \leq \mathcal F^t p(0)$. This inequality shows 
\begin{equation}\label{eq:r<=rho(F)}
\alpha \leq \rho(\mathcal F), 
\end{equation}
where $\rho(\cdot)$ denotes the spectral radius of a matrix. 

\begin{figure}[tb]
\centering
\includegraphics[width=.475\linewidth]{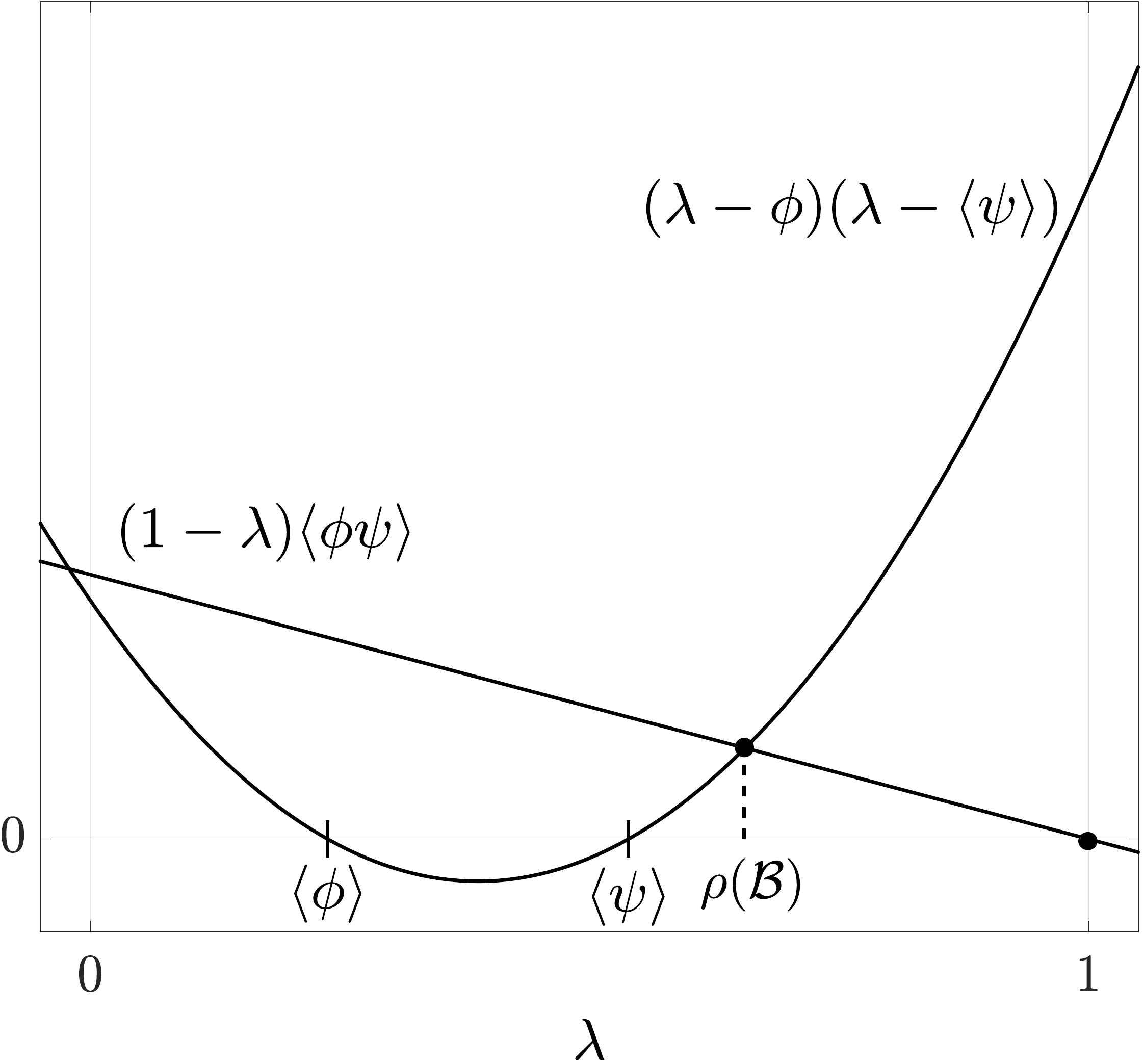}
\caption{Characteristic equation \cref{eq:charEquation}}
\label{fig:illustration}
\end{figure}

Now, we evaluate $\rho(\mathcal F)$. Equation~\cref{eq:def:calF} is rewritten as $\mathcal F = (1-\delta) I + \beta \mathcal A $, where $\mathcal A = \onev \onev^\top - (\onev - \psi)( \onev - \phi)^\top$. Since $\mathcal A$ is nonnegative entrywise and $1-\delta \geq 0$, we obtain
\begin{equation}\label{eq:rho(F)=}
\rho(\mathcal F) = 1 -\delta + \beta \rho(\mathcal A). 
\end{equation}
Furthermore, as we prove in \cref{app:rhoA}, it holds that 
\begin{equation}\label{eq:rho(A)=rho(nB)}
\rho(\mathcal A) = \rho 
(n\mathcal B), 
\end{equation}
where
\begin{gather}\label{eq:def:langphipsirnv}
\mathcal B = \begin{bmatrix}
1 & 1 - \av \psi
\\
-1+ \av \phi& 
-1+\av \phi + \av \psi - \av {\phi\psi}
\end{bmatrix}, 
\\
\label{eq:def:langphipsirang}
\langle\phi \rangle = \frac 1 n \sum_{i=1}^n \phi_i
, \quad 
\langle\psi \rangle = \frac 1 n \sum_{i=1}^n \psi_i
, \quad 
\langle\phi\psi \rangle = \frac 1 n \sum_{i=1}^n \phi_i\psi_i.  
\end{gather}
As shown in \cref{fig:illustration}, matrix $\mathcal B$ has the characteristic equation 
\begin{equation}\label{eq:charEquation}
(1-\lambda) \av{\phi\psi} = ( \lambda - \av \phi)( \lambda - \av \psi) 
\end{equation}
having the roots 
\begin{equation}\label{eq:roots}
\lambda = \frac{\av \phi + \av \psi - \av{\phi\psi} \pm \sqrt{
{(\av \phi + \av \psi - \av{\phi\psi})}^2 + 4(\av{\phi\psi} - \av \phi \av \psi) 
}}{2}. 
\end{equation}
The roots are real because 
\begin{equation}\label{eq:suportRealness}
{(\av \phi + \av \psi - \av{\phi\psi})}^2 + 4(\av{\phi\psi} - \av \phi \av \psi) 
\geq 
{(\av \psi - \av \phi)}^2 + \av{\phi\psi}^2 
> 0, 
\end{equation}
which follows from the trivial inequality~$4\av{\phi\psi} \geq 2\av \phi \av{\phi\psi} + 2\av \psi \av{\phi\psi}$. Therefore, by substituting  equation~\cref{eq:def:phipsi} into equation~\cref{eq:roots}, we obtain $\rho(\mathcal B) = \kappa \bar m $. This equation and \cref{eq:r<=rho(F),eq:rho(A)=rho(nB),eq:rho(F)=} complete the proof of the \lcnamecref{thm:eqanalysis}.
\end{proof}

The following corollary shows that an epidemic will become extinct more quickly when the adaptation factor and acceptance rate are less correlated in a weighted sense. 

\begin{corollary}\label{cor:sensitivity}
Let $(\chi, \pi)$ and $(\chi', \pi')$ be pairs of adaptation factors and acceptance rates of nodes, and denote the corresponding upper-bounds on the decay rates by $\alpha_{\rm u}$ and $\alpha_{\rm u}'$, respectively. If $\av{\chi}_{\!a} =  \av{\chi'}_{\!a}$, $\av{\pi}_{\!a} =  \av{\pi'}_{\!a}$, and $\av{\chi\pi}_{\!a^2} <  \av{\chi'\pi'}_{\!a^2}$, then 
\begin{equation}
\alpha_{\rm u} < \alpha_{\rm u}'. 
\end{equation}
\end{corollary}

\begin{proof}
By the proof of \cref{thm:eqanalysis}, we have $\alpha_{\rm u} = 1 - \delta + \rho(\mathcal B) n\beta$. \cref{fig:illustration} implies that $\rho(\mathcal B)$ increases with $\langle \phi \psi \rangle$ when $\av \phi$ and $\av \psi$ are fixed. This proves the claim of the corollary because $\av{\chi\pi}_{\!a^2} = \langle \phi \psi \rangle/\bar m^2$, $\av{\chi}_{\!a} = \av \phi/\bar m$, and $\av{\pi}_{\!a} = \av \psi/\bar m$.
\end{proof}

As another corollary of \cref{thm:eqanalysis}, we also present an upper bound on the decay rate when nodes do not adapt to the states of the nodes. 

\begin{corollary}\label{cor:SIS}
Assume $\chi_i = \pi_i = 1$ for all $i$. Let 
\begin{equation}
\kappa_0 = 
\frac{2\av{a} - \bar m\av{a^2} + \sqrt{
4\av{a^2} - 4 \bar m \av{a} \av{a^2} +\bar m^2\av{a^2}^2 
}}{2}. 
\end{equation}
Then, the decay rate of the activity-driven SIS model is at most 
$1 - \delta + \kappa_0 \bar mn \beta$. 
\end{corollary}

\begin{remark}
If $m$ is sufficiently small compared with $n$ and, furthermore, $n$ is sufficiently large (as implicitly assumed in~\cite{Perra2012}), the upper-bound in \cref{cor:SIS} reduces to $1 - \delta + (\av a + \sqrt{\av{a^2}})m\beta$, which coincides with the result in~\cite{Perra2012}.
\end{remark}

\section{Cost-optimal adaptations}\label{sec:optimizaiton}

In this section, we study the problem of eradicating an epidemic outbreak by distributing resources to nodes in the activity-driven network. We consider the situation in which there is a budget that can be invested on strengthening the preventative behaviors of each node. We show that the optimal budget allocation is found using geometric programs, which can be efficiently solved in polynomial time.

\subsection{Problem statement}

We consider an optimal resource allocation problem in which we can tune the adaptation factors and acceptance rates of nodes. Assume that, to set the adaptation factor of node $v_i$ to $\af_i$, we need to pay a cost~$\costAF_i(\af_i)$. Similarly we need to pay a cost~$\costAP_i(\ap_i)$ to set the acceptance rate of node $v_i$ to $\ap_i$. The total cost for tuning the parameters to the values $\af_1$, \dots, $\af_n$, $\ap_1$, \dots, $\ap_n$ equals
\begin{equation}
C = \sum_{i=1}^n (\costAF_i(\af_i) + \costAP_i(\ap_i)). 
\end{equation}
Throughout this section, we assume the following box constraints: 
\begin{equation}\label{eq:boxConstraints} 
0 < \ubar \af_i \leq \af_i \leq \bar \af_i
,\quad
 0< \ubar \ap_i \leq \ap_i \leq \bar \ap_i.   
\end{equation}

In this paper, we consider the following two types of optimal resource allocation problems. 

\begin{problem}[Cost-constrained optimal resource allocation]\label{prb:}
Given a total budget $\bar C$, find the adaptation rates and acceptance rates that minimize $\alpha_{\rm u}$ while satisfying the budget constraint
\begin{equation}\label{eq:budgetConstraint}
C \leq \bar C. 
\end{equation}
\end{problem}

\begin{problem}[Performance-constrained optimal resource allocation]\label{prb:pc}
Given a largest tolerable decay rate $\bar \alpha$, find the adaptation rates and acceptance rates that minimize the total cost $C$ while satisfying the performance  constraint
\begin{equation}\label{eq:performanceConstraint}
\alpha_{\rm u}\leq \bar \alpha. 
\end{equation}
\end{problem}

\subsection{Cost-constrained optimal resource allocation}

In this subsection, we show that \cref{prb:} can be transformed to a geometric program~\cite{Boyd2007}, which can be efficiently solved. Before stating our main results, we give a brief review of geometric programs. Let $x_1$, \dots, $x_n$ denote positive variables and define $x = (x_1, \dotsc, x_n)$. We say that a real function~$q(x)$ is a \emph{monomial} if there exist $c \geq 0$ and $a_1, \dotsc, a_n \in \mathbb{R}$ such that $q(x) = c x_{\mathstrut 1}^{a_{1}} \dotsm x_{\mathstrut n}^{a_n}$. Also, we say that a function~$r(x)$ is a \emph{posynomial} if it is a sum of monomials of~$x$ (we point the readers to~\cite{Boyd2007} for more details). Given a collection of posynomials $r_0(x)$, \dots, $r_k(x)$ and monomials $q_1(x)$, \dots, $q_\ell(x)$, the optimization problem
\begin{equation} 
\begin{aligned}
\minimize\ \ \ \, 
&
r_0(x)
\\
\st\ \ 
&
r_i(x)\leq 1,\quad i=1, \dotsc, k, 
\\
&
q_j(x) = 1,\quad j=1, \dotsc, \ell, 
\end{aligned}
\end{equation}
is called a \emph{geometric program}. A constraint of the form $r(x)\leq 1$ with $r(x)$ being a posynomial is called a \emph{posynomial constraint}.  Although geometric programs are not convex, they can be efficiently converted into equivalent convex optimization problems~\cite{Boyd2007}.

We assume that the cost functions~$\costAF_i$ and $\costAP_i$ decrease with the adaptation factor~$\af_i$ and acceptance rate~$\ap_i$, respectively. This assumption implies a natural situation in which it is more costly to suppress $\chi_i$ and $\pi_i$ to a larger extent. We also expect diminishing returns with increasing investments~\cite{Reluga2010}. For a fixed $\epsilon > 0$, let $\Delta f_i(\chi_i) = f_i(\chi_i-\epsilon) - f_i(\chi_i)$ denote the cost for improving the adaptation factor from $\chi_i$ to~\mbox{$\chi_i -\epsilon$}. Then,  diminishing returns imply that $\Delta f_i$  decreases with $\chi_i$, which implies the convexity of $f_i$. Therefore, we place the following assumption on the cost functions.
\begin{assumption}\label{assm:}
For all $i \in \{1, \dotsc, n\}$, decompose $f_i$ and $g_i$ into the differences of their positive and negative parts as follows:
\begin{align}
f_i &= f_i^+ - f_i^-, 
\\
g_i &= g_i^+ - g_i^-, 
\end{align}
where $f_i^+ = \max(f, 0)$, $f_i^- = \max(-f, 0)$, $g_i^+ = \max(g, 0)$, and $g_i^- = \max(-g, 0)$. Then, $f_i^+$ and $g_i^+$ are posynomials, and $f_i^-$ and $g_i^-$ are constants.
\end{assumption}

\cref{assm:} allows us to use any cost functions that are convex on the log-log scale because any function convex on the log-log scale can be approximated by a posynomial with an arbitrary accuracy~\cite[Section~8]{Boyd2007}. We now state our first main result in this section, which allows us to efficiently solve \cref{prb:} via geometric programming:

\begin{theorem}\label{thm:bc}
Let $\af_i^\star$ and $\ap_i^\star$ be the solutions of the following optimization problem: 
\begin{subequations}\label{eq:opt}
\begin{align}
\minimize_{\tilde \lambda,\, {\af_{i}},\, \ap_i,\,\zeta,\,\eta > 0}\ \ & 1/\tilde \lambda
\\
\st\ \ \ &\mbox{\cref{eq:boxConstraints},}
\\
& \bar m^2 \tilde \lambda \av{\af \ap}_{\!a^2}\zeta\eta  \leq 1, 
\label{eq:quadConstraint}
\\
& \zeta^{-1} + \tilde \lambda + \bar m \av{\af}_{\!a}\leq 1,  \label{eq:zetaConstraint}
\\
& \eta^{-1} + \tilde \lambda +  \bar m \av{\ap }_{\!a}\leq 1,  \label{eq:etaConstraint}
\\
&\sum_{i=1}^n (\costAF_i^+(\af_i) + \costAP_i^+(\ap_i)) \leq \bar C + \sum_{i=1}^n (\costAF_{i}^- + \costAP_{i}^-).\label{eq:tildeCost<barC}
\end{align}
\end{subequations}
Then, the adaptation factor $\af_i = \af_i^\star$ and the acceptance rate $\ap_i = \ap_i^\star$ solve \cref{prb:}. Moreover, under \cref{assm:}, the optimization problem \cref{eq:opt} is a geometric program.
\end{theorem}

To prove this theorem, we show an alternative characterization of the decay rate in terms of inequalities. 

\begin{lemma}\label{lem:ineqanalysis}
Let $\lambda > 0$. The upper bound $\alpha_{\rm u}$ satisfies 
\begin{equation}\label{eq:ineqChar}
\alpha_{\rm u} \leq 1-\delta + \lambda n \beta
\end{equation}
if and only if
\begin{align}
(1-\lambda) \av{\phi\psi} &\leq ( \lambda - \av \phi)( \lambda - \av \psi), 
\label{eq:lambdaIneq1}
\\
\av{\phi} &< \lambda,
\label{eq:lambdaIneq2}
\\
\av{\psi} &< \lambda.
\label{eq:lambdaIneq3}
\end{align}
\end{lemma}

\begin{proof}
By the proof of \cref{thm:eqanalysis}, inequality~\cref{eq:ineqChar} holds true if and only if $\lambda \geq \rho(\mathcal B)$. \cref{fig:illustration} indicates that $\lambda \geq \rho(\mathcal B)$ is equivalent to conditions \cref{eq:lambdaIneq1,eq:lambdaIneq2,eq:lambdaIneq3}.
\end{proof}

We can now prove \cref{thm:bc}: 

\begin{proof}[Proof of \cref{thm:bc}]
By \cref{lem:ineqanalysis}, the solutions of \cref{prb:} are given by those of the following optimization problem:
\begin{subequations}\label{eq:optpre}
\begin{align}
\minimize_{\lambda,\, {\af_{i}},\, \ap_i > 0}\ \ \ \, &    1-\delta +  \lambda n \beta 
\\
\st\ \ &\mbox{\cref{eq:lambdaIneq1,eq:lambdaIneq2,eq:lambdaIneq3,eq:boxConstraints,eq:budgetConstraint}}. 
\end{align}
\end{subequations}
Define the auxiliary variables $\zeta = {1}/{(\lambda - \av \phi)}$ and $\eta = {1}/{(\lambda - \av \psi)}$. Then, conditions~\cref{eq:lambdaIneq1,eq:lambdaIneq2,eq:lambdaIneq3} hold true if and only if $(1-\lambda)\av{\phi\psi}\zeta\eta \leq 1$, $\zeta > 0$, and $\eta > 0$. Therefore, the optimization problem~\cref{eq:optpre} is equivalent to the following optimization problem:
\begin{subequations}\label{eq:opt1}
\begin{align}
\minimize_{\lambda,\, {\af_{i}},\, \ap_i,\,\zeta,\,\eta > 0}\ \ &   \lambda
\\
\st\ \ \ &\mbox{\cref{eq:boxConstraints,eq:budgetConstraint}}, 
\\& (1-\lambda) \av{\phi\psi} \zeta\eta \leq 1, 
\\
& \zeta^{-1}  - \lambda + \av{\phi}= 0, 
 \label{eq:pre:zetaConst}
\\
& \eta^{-1} - \lambda + \av{\psi} = 0 , 
 \label{eq:pre:etaConst}
\end{align}
\end{subequations}
where we minimize $\lambda$ instead of $1-\delta + \lambda n \beta$. We claim that the optimal value of the objective function is equal to the one in the following optimization problem: 
\begin{subequations}\label{eq:opt2}
\begin{align}
\minimize_{\lambda,\, {\af_{i}},\, \ap_i,\,\zeta,\,\eta > 0}\ \ &  \lambda
\\
\st\ \ \  &\mbox{\cref{eq:boxConstraints,eq:budgetConstraint}}, 
\\ &(1-\lambda) \av{\phi\psi} \zeta\eta \leq 1, 
\label{eq:pre:quadConstraint}
\\
& \zeta^{-1} - \lambda + \av{\phi} \leq 0, 
\label{eq:prepre:zetaConst}
\\
& \eta^{-1} - \lambda  + \av{\psi}  \leq 0. 
\label{eq:prepre:etaConst}
\end{align}
\end{subequations}

Let $\lambda_1^\star$ and $\lambda_2^\star$ be the optimal values of the objective functions in problems \cref{eq:opt1} and \cref{eq:opt2}, respectively. We have $\lambda_1^\star \geq \lambda_2^\star$ because the constraints in problem \cref{eq:opt1} are more strict than those in \cref{eq:opt2}. Let us show $\lambda_1^\star \leq \lambda_2^\star$. Assume that the optimal value~$\lambda_2^\star$ in problem~\cref{eq:opt2} is attained by the parameters $(\lambda, \af_i, \ap_i, \zeta, \eta) = (\lambda^\star, \af_i^\star, \ap_i^\star, \zeta^\star, \eta^\star)$.  Since the left-hand sides of constraints \cref{eq:prepre:etaConst,eq:prepre:zetaConst} decease with~$\zeta$ and $\eta$, there exist nonnegative constants $\Delta \zeta$ and $\Delta \eta$ such that $\zeta = \zeta^\star - \Delta \zeta$ and $\eta = \eta^\star - \Delta \eta$ satisfy the equality constraints \cref{eq:pre:etaConst,eq:pre:zetaConst}. Moreover, since the left-hand side of the constraint~\cref{eq:pre:quadConstraint} increases with~$\zeta$ and~$\eta$, the new set of parameters $(\lambda, \af_i, \ap_i, \zeta, \eta) = (\lambda^\star, \af_i^\star, \ap_i^\star, \zeta^\star - \Delta \zeta, \eta^\star - \Delta \eta)$ still satisfies \cref{eq:pre:quadConstraint}. Furthermore, these changes of parameters do not affect the feasibility of the box constraints \cref{eq:boxConstraints} and the budget constraint~\cref{eq:budgetConstraint} because the constraints are independent of the values of $\zeta$ and $\eta$. Therefore, we have shown the existence of parameters achieving $\lambda = \lambda_2^\star$ but still satisfying the constraints in the optimization problem~\cref{eq:opt1}. This shows $\lambda_1^\star \leq \lambda_2^\star$, as desired.

Now, by rewriting the optimization problem \cref{eq:opt2} in terms of the variables $\tilde \lambda = 1-\lambda$ and substituting \cref{eq:def:phipsi} in  \cref{eq:opt2}, we obtain the optimization problem~\cref{eq:opt}. Notice that minimizing $\lambda$ is equivalent to maximizing $1-\tilde \lambda$, which is equivalent to minimizing $1/\tilde \lambda$.

Let us finally show that \cref{eq:opt} is a geometric program. The objective function, $1/\tilde \lambda$, is a posynomial in $\tilde \lambda$. The constraints \cref{eq:quadConstraint,eq:etaConstraint,eq:zetaConstraint,eq:boxConstraints} are posynomial constraints. Finally, \cref{assm:} guarantees that  constraint~\cref{eq:tildeCost<barC} is a posynomial constraint as well. This completes the proof of the \lcnamecref{thm:bc}.
\end{proof}

\begin{figure}[tb]
\centering
\includegraphics[width=1\linewidth,trim={0in 4.5in 0in 0in},clip]{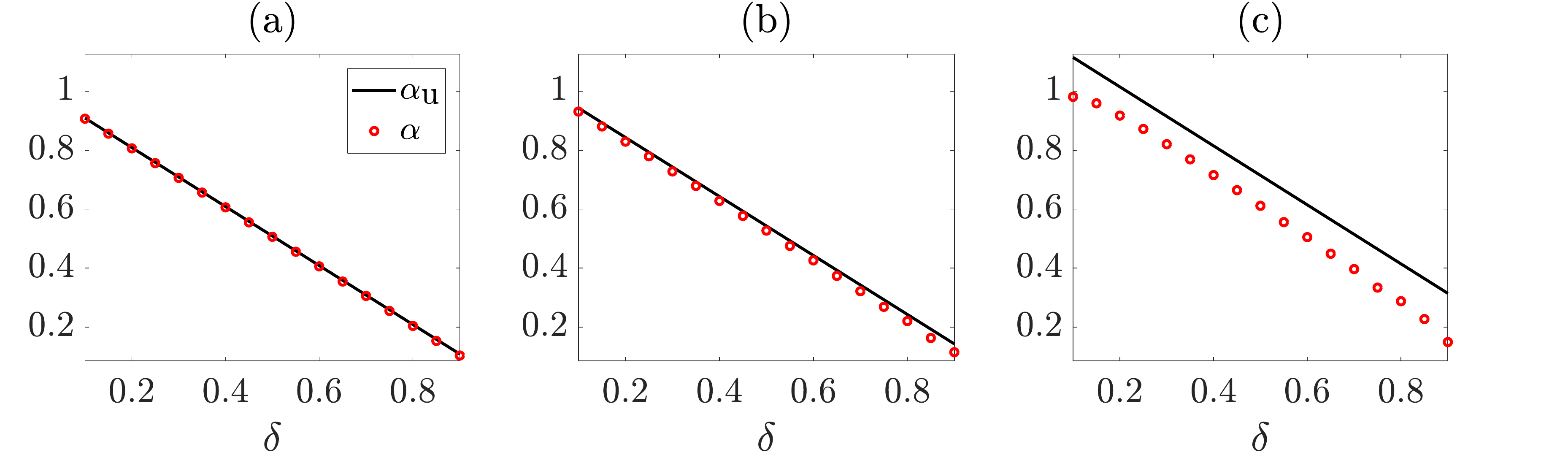}
\caption{Comparison between the numerically obtained decay rates $\alpha$ and their upper bounds $\alpha_{\rm u}$ in \labelcref{case:uniform} when $\beta = 0.8$. (a) $m=2$, (b) $m=10$, and (c) $m=50$.}
\label{fig:analysisUniformComparison}
\centering
\includegraphics[width=1\linewidth,trim={0in 0in 0in 4.45in in},clip]{fig4.pdf}
\caption{Discrepancy between the true decay rates and their upper bounds in \labelcref{case:uniform}. (a)~$m=2$, (b) $m=10$, and (c) $m=50$.}
\label{fig:analysisUniformErrors}
\end{figure}

\begin{figure}[tb]
\centering
\includegraphics[width=1\linewidth,trim={0in 4.5in 0in 0in},clip]{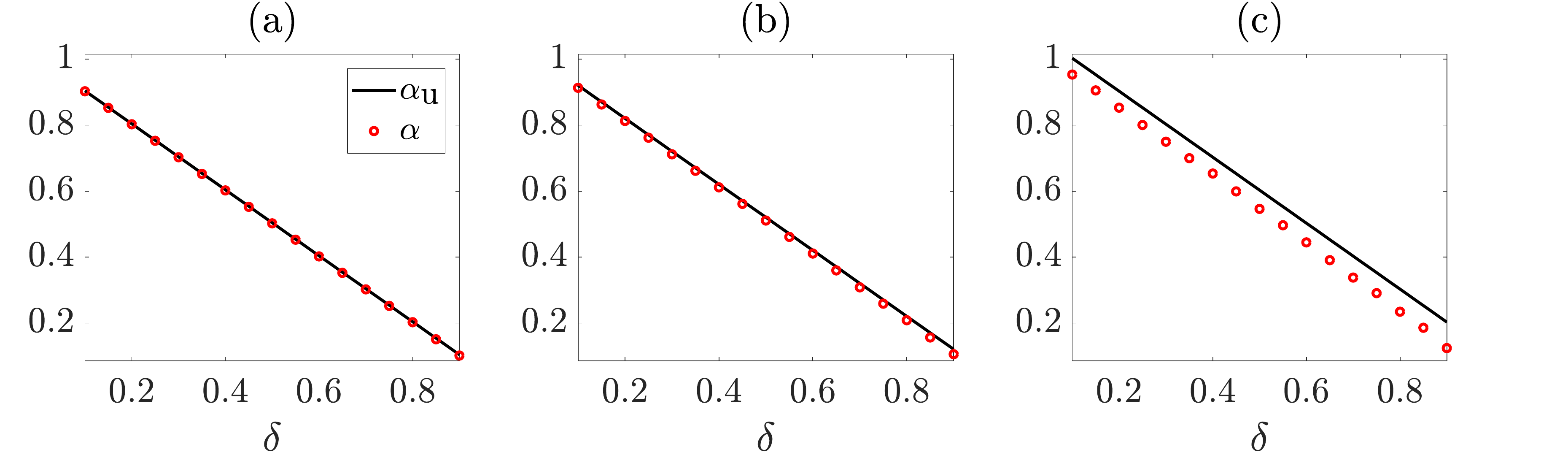}
\caption{Comparison between the numerically obtained decay rates $\alpha$ and their upper bounds $\alpha_{\rm u}$ in \labelcref{case:power} when $\beta = 0.8$. (a) $m=2$, (b) $m=10$, and (c) $m=50$.}
\label{fig:analysisPowerComparison}
\centering
\includegraphics[width=1\linewidth,trim={0in 0in 0in 4.45in in},clip]{fig5.pdf}
\caption{Discrepancy between the true decay rates and their upper bounds in \labelcref{case:power}. (a)~$m=2$, (b) $m=10$, and (c) $m=50$.}
\label{fig:analysisPowerErrors}
\end{figure}

\subsection{Performance-constrained optimal resource allocation}

In the same way as in the previous section, we can efficiently solve \cref{prb:pc} via geometric programming: 

\begin{theorem}\label{thm:pc}
Let $\af_i^\star$ and $\ap_i^\star$ be the solution of the following optimization problem: 
\begin{subequations}
\begin{align}
\minimize_{\tilde \lambda,\, {\af_{i}},\, \ap_i,\,\zeta,\,\eta > 0}\ \ & \sum_{i=1}^n (\costAF_i^+(\af_i) + \costAP_i^+(\ap_i))
\\
\st\ \ \  & 
\mbox{\upshape\cref{eq:quadConstraint,eq:etaConstraint,eq:zetaConstraint,eq:boxConstraints}}, 
\\
& \frac{\beta n + 1 - \delta - \bar \alpha}{\beta n}\tilde \lambda^{-1} \leq 1.\label{eq:tildeperformance<bargamma}
\end{align}
\end{subequations}
Then, the adaptation factor $\af_i = \af_i^\star$ and the acceptance rate $\ap_i = \ap_i^\star$ solve \cref{prb:pc}. Moreover, under \cref{assm:}, the optimization problem is a geometric program.
\end{theorem}

\begin{proof}
Constraint~\cref{eq:tildeperformance<bargamma} is equivalent to the performance constraint~\cref{eq:performanceConstraint}. The rest of the proof is almost the same as the proof of \cref{thm:bc} and is omitted. 
\end{proof}

\section{Numerical simulations}\label{ref:numerical}

In this section, we illustrate the theoretical results obtained in previous sections by numerical simulations.

\begin{figure}[tb]
\centering
\includegraphics[width=.55\linewidth]{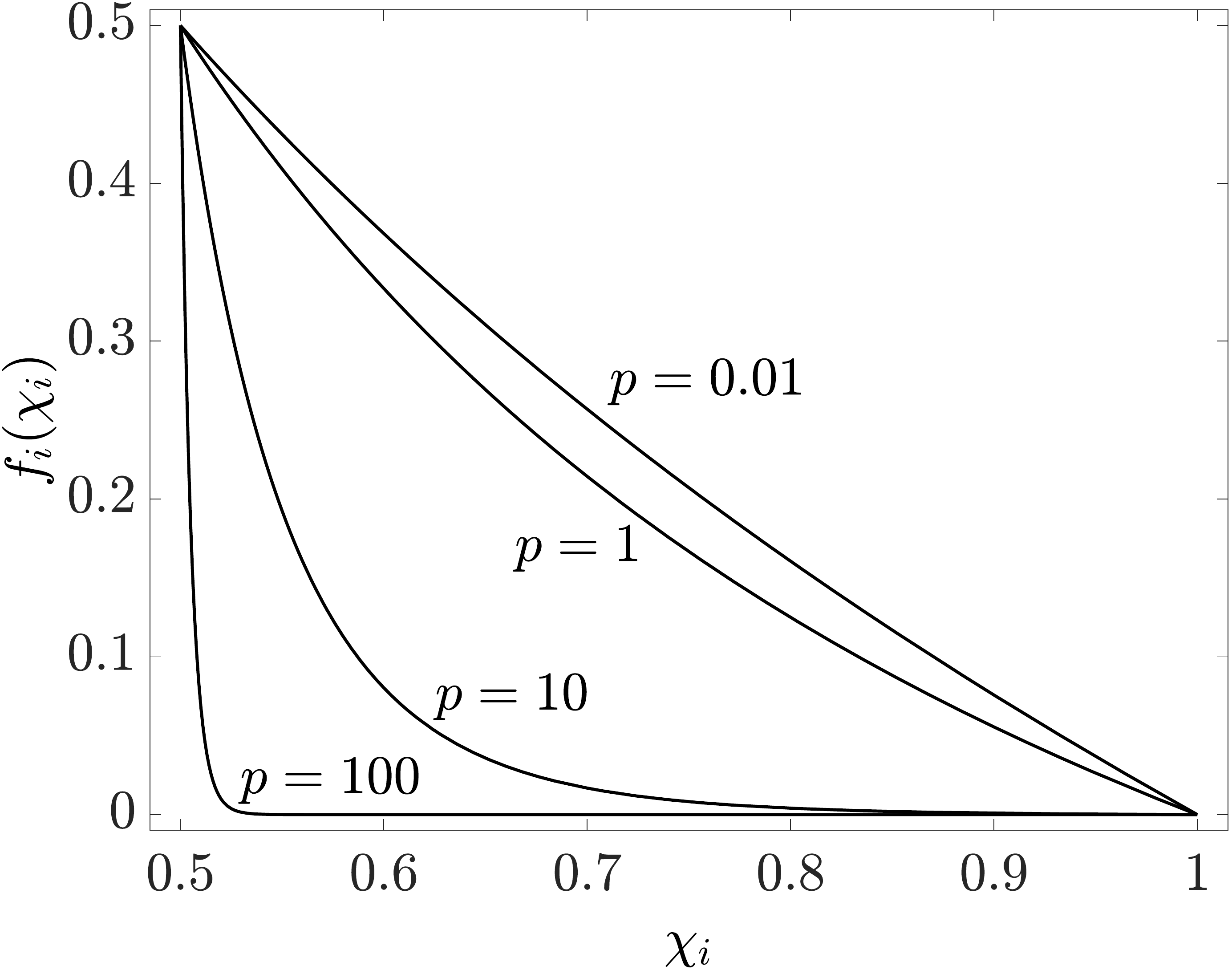}
\caption{Cost function $\costAF_i(\af_i)$ for $p=0.01$, $1$, $10$, and $100$ when $\myubar{$\af$}{1.5pt}_i=0.5$.}
\label{fig:costFunctions}
\end{figure}

\subsection{Accuracy of the upper bound}

We first illustrate the accuracy of the upper bound \cref{eq:upperBound} on the decay rate. We use an activity-driven network with $n=250$ nodes and study the following two cases:
\begin{enumerate}[labelindent=\parindent, leftmargin=*, label=Case \arabic*., ref=Case \arabic*]
\item\label{case:uniform} Activity rates following a uniform distribution over $[0, 10^{-2}]$; 
\item\label{case:power} Activity rates following a probability distribution $F(a)$ that is proportional to $a^{-2.8}$ in the interval~$[10^{-3}, 1]$ and equal to zero elsewhere~\cite{Perra2012}. 
\end{enumerate}
We assume that both the adaptation factors~$\af_i$ and acceptance rates~$\ap_i$ follow a uniform distribution over $[0, 1]$. For various values of the infection rate $\beta$, the recovery rate~$\delta$, and $m$, we compute the decay rate~$\alpha$ based on numerical simulations of the model and its upper bound~$\alpha_{\rm u}$. To compute the decay rate~$\alpha$, we use Monte Carlo simulations. For each triple $(\beta, \delta, m)$, we run \num[group-separator={,}]{10000} simulations of the activity-driven A-SIS model with all nodes being infected at time $t=0$. In each numerical simulation, we compute the probability vector~$p(t)$ for each $t=0$, $1$, $2$, $\dotsc$ until the norm $\lVert p(t) \rVert$ falls below $0.1$. We then estimate the decay rate by $\alpha = \max_{t}t^{-1}\log \lVert p(t) \rVert$.

In \cref{fig:analysisUniformComparison}, we let $\beta = 0.8$ and compare the decay rates and their upper bounds in \labelcref{case:uniform} for various values of $\delta$ and $m$. We confirm that $\alpha_{\rm u}$ bounds the numerically obtained decay rates. The discrepancy $\alpha_{\rm u}-\alpha$ increases with $m$. To further examine how the discrepancy depends on the parameters, we present the discrepancy for various values of $\beta$, $\delta$, and $m$ in  \cref{fig:analysisUniformErrors}. Besides the aforementioned dependence of the discrepancy on $m$, we also see that the discrepancy tends to be large when $\beta$ is large. We observe the same trend in the case of the power-law distribution of the activity rate (\labelcref{case:power}; shown in  \cref{fig:analysisPowerComparison,fig:analysisPowerErrors}).

\subsection{Optimal resource distribution}

\newcommand{\ubarpi}{\myubar{$\ap$}{.25pt}}

\begin{figure}[tb]
\centering
\includegraphics[width=1\linewidth]{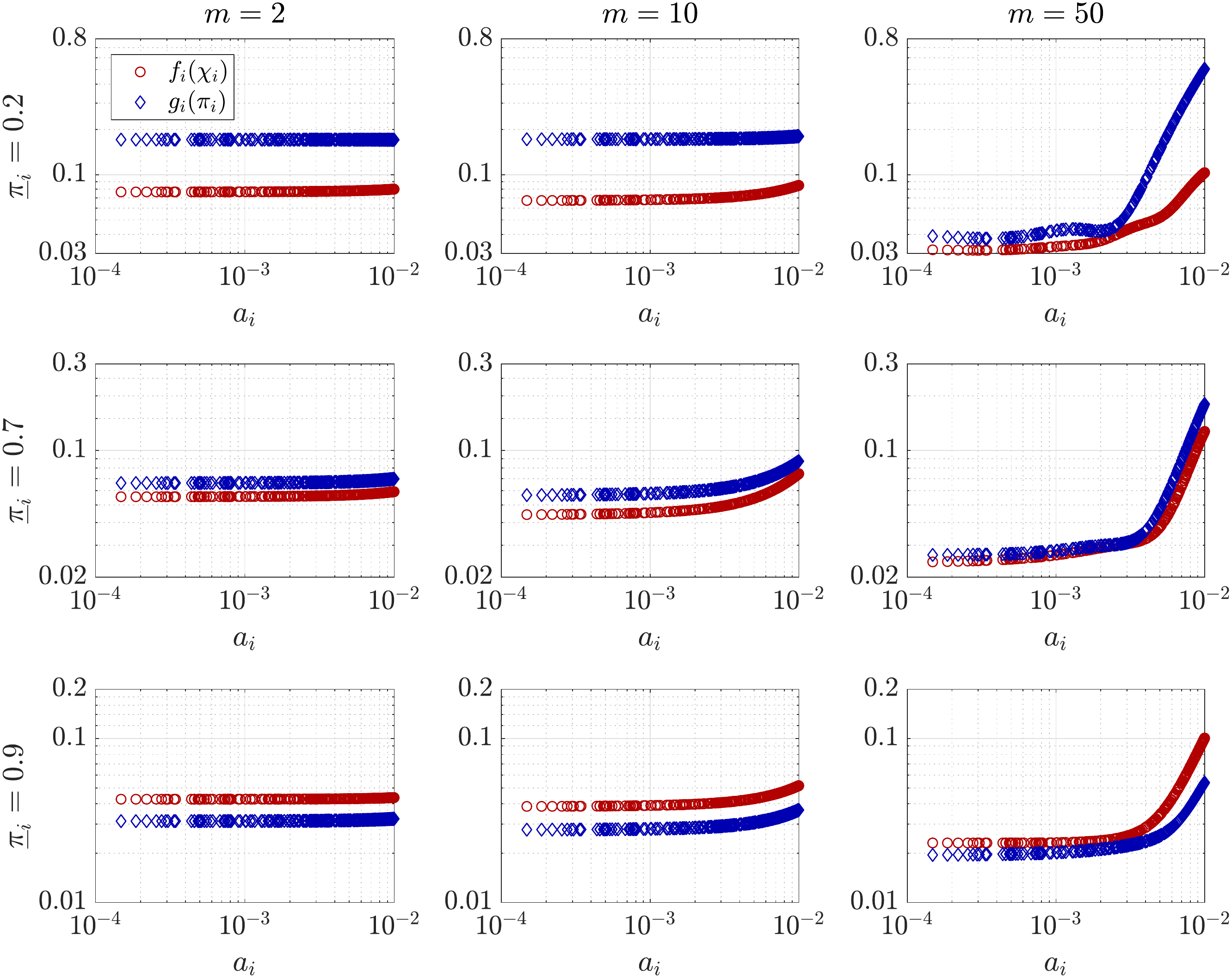}
\caption{Optimal investments on adaptation factors (circles) and acceptance rates (squares) in \labelcref{case:uniform}. Left column: $m=2$, middle column: $m=10$, right column: $m=50$. Top row: $\ubarpi_i = 0.2$, middle row: $\ubarpi_i = 0.7$, bottom row: $\ubarpi_i = 0.9$.}
\label{fig:bc:random}
\end{figure}

We numerically illustrate our framework to solve the optimal resource allocation problems developed in \Cref{sec:optimizaiton}. We assume $\bar \af_i = 1$ and $\bar \ap_i = 1$ in the box constraints~\cref{eq:boxConstraints}. We  use the following cost functions (similar to the ones in~\cite{Preciado2014}):
\begin{equation}\label{eq:costFunctions}
\costAF_i(\af_i) =  (1-\ubar \af_i)\frac{\af_i^{-p} - 1}{{\ubar \af}_i^{-p} - 1},\quad 
\costAP_i(\ap_i) =  (1-\ubar \ap_i)\frac{\ap_i^{-q} - 1}{{\ubar \ap}_i^{-q} - 1}. 
\end{equation}
These cost functions satisfy \cref{assm:}. Parameters $p, q>0$ tune the shape of the cost functions as illustrated in \cref{fig:costFunctions}. Because the cost functions are normalized as $\costAF_i(\ubar \af_i)=1-\ubar \af_i$, $\costAF_i(\bar \af_i)=\costAF_i(1)=0$, $\costAP_i(\ubar \ap_i)=1-\ubar \ap_i$, and $\costAP_i(\bar \ap_i)=\costAP_i(1)=0$, the maximum adaptation $(\af_i, \ap_i) = (\ubar \af_i, \ubar \ap_i)$ ($1\leq i \leq n$) in the network is achieved with the budget
\begin{equation}
C_{\max} = 2n - \sum_{i=1}^n (\ubar \af_i + \ubar \ap_i). 
\end{equation}

As in our previous simulations, we consider the activity-driven A-SIS model over a network with $n=250$ nodes. We let $m$ be either $2$, $10$, or $50$. Let $\ubar \af_i = 0.8$. We let the value of $\ubar \ap_i$ be either $0.2$, $0.7$, or $0.9$, and use $p=q=0.01$ for the cost functions~\cref{eq:costFunctions}. We use the fixed budget $\bar C = C_{\max}/4$. For each pair $(m, \ubar \pi_i)$, we determinate the adaptation factors and acceptance rates for the cost-constrained optimal resource allocation problem (\cref{prb:}) by solving the geometric program shown in \cref{thm:bc}.

\begin{figure}[tb]
\centering \includegraphics[width=1\linewidth]{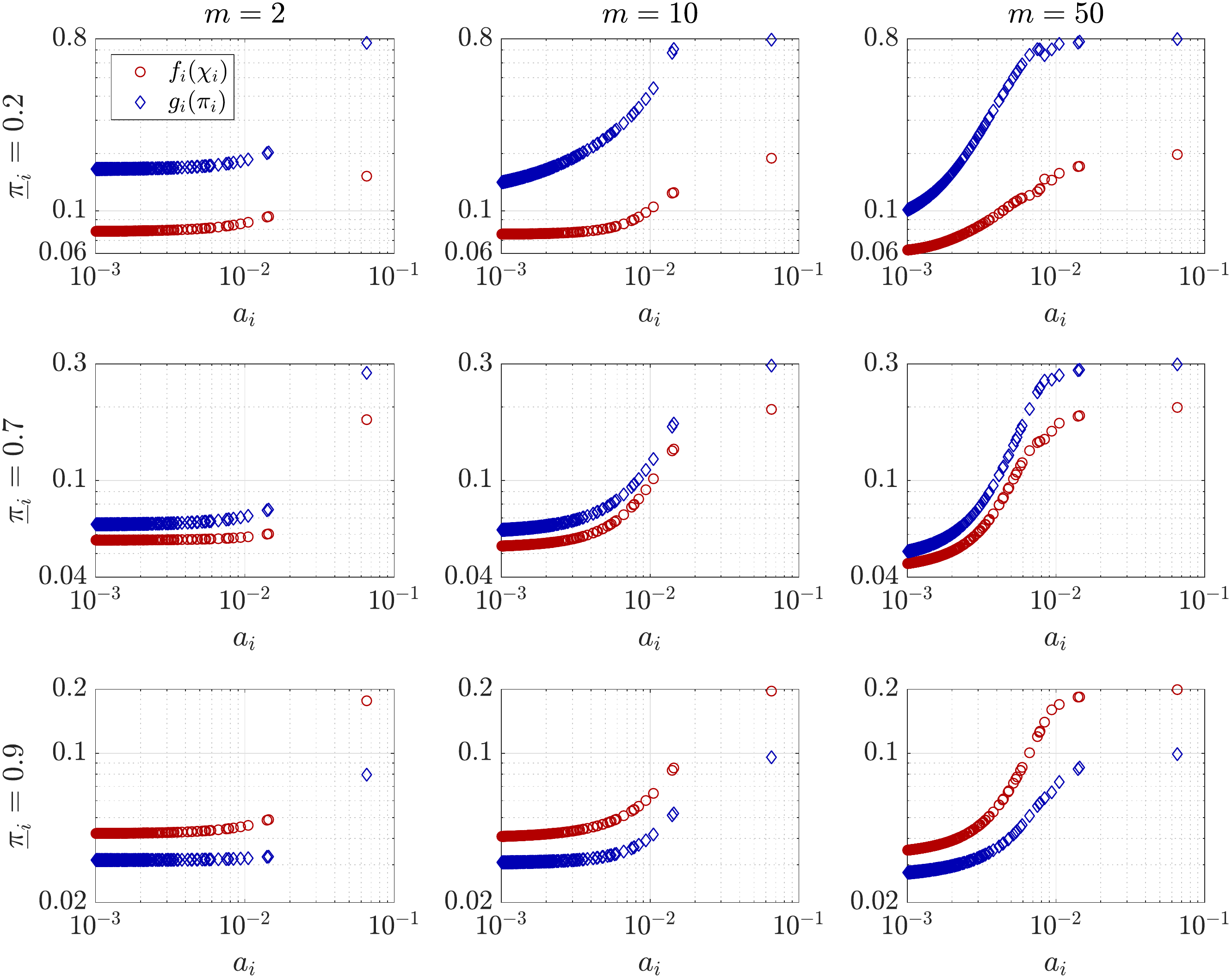} \caption{Optimal investments on adaptation factors (red circles) and acceptance rates (blue squares) in \labelcref{case:power}. Left column: $m = 2$, middle column: $m = 10$, right column: $m = 50$. Top row: $\ubarpi_i = 0.2$, middle row: $\ubarpi_i = 0.7$, bottom row: $\ubarpi_i = 0.9$.}\label{fig:bcpower} \end{figure}

The optimal investments on the adaptation factors and acceptance rates (i.e., $\costAF_i(\af_i)$ and $\costAP_i(\ap_i)$) are shown in \cref{fig:bc:random} for \labelcref{case:uniform}. We see that the smaller the lower limit of the acceptance rate $\ubar \ap_i$, the more we should invest on decreasing the acceptance rates. We also find that, in the case of $m=2$ and $10$, the resulting investment is almost independent of the activity rates of nodes. This trend disappears for larger values of~$m$. In the case of $m=50$, the optimal solution disproportionately invests on the nodes having high activity rates. This result reflects the structure of the optimization problem~\cref{eq:opt} for the following reason. If $m \ll n$, then $\bar m \ll 1$. Therefore, the set of constraints in the optimization problem~\cref{eq:opt} approximately reduces to the set of constraints~\cref{eq:boxConstraints},
\begin{gather}
 \zeta^{-1} + \tilde \lambda \leq 1,
\\
 \eta^{-1} + \tilde \lambda \leq 1,  
\end{gather}
and \cref{eq:tildeCost<barC}. 
These four constraints do not involve activity rates, which leads to optimal adaptation factors and acceptance rates that are independent of the activity rates. On the other hand, for a large $m$,  constraints~\cref{eq:quadConstraint,eq:zetaConstraint,eq:etaConstraint} involving the weighted sums~$\av{\chi}_{\!a}$, $\av{\pi}_{\!a}$, and $\av{\chi \pi}_{\!a^2}$ become tighter, rendering investments on high-activity nodes more effective.

The optimal investments for \labelcref{case:power} are shown in \cref{fig:bcpower}. As in \labelcref{case:uniform}, the smaller $\ubar \pi_i$, the optimal solution invests more on decreasing the acceptance rates. However, the dependence of the optimal solution on the value of $m$ is not as strong as in \labelcref{case:uniform}.

\section{Conclusion}

In this paper, we have studied epidemic processes taking place in temporal and adaptive networks. Based on the activity-driven network model, we have proposed the activity-driven A-SIS model, where infected individuals adaptively decrease their activity and reduce connectivity with other nodes to prevent the spread of the infection. In order to avoid the computational complexity arising from the model, we have first derived a linear dynamics able to upper-bound the infection probabilities of the nodes. We have then derived an upper-bound on the decay rate of the expected number of infected nodes in the network in terms of the eigenvalues of a $2\times 2$ matrix. Then, we have shown that a small correlation between the two different adaptation mechanisms is desirable for suppressing epidemic infection. Furthermore, we have proposed an efficient algorithm to optimally tune the adaptation rates in order to suppress the number of infected nodes in networks. We have illustrated our results by numerical simulations.

\appendix
\section{Proof of \cref{eq:rho(A)=rho(nB)}}\label{app:rhoA}
The image of $\mathcal A$ is spanned by $\onev$ and $\onev-\psi$. All but two eigenvalues of $\mathcal A$ are zero. Therefore, an eigenvector of $\mathcal A$ corresponding to the spectral radius of $\mathcal A$ is a linear combination of vectors $\onev$ and  $\onev-\psi$. Assume that $v = s_1 \onev + s_2(\onev -\psi)$ is such an eigenvector of $\mathcal A$ with eigenvalue~$\lambda$. By comparing the coefficients of the vectors $\onev$ and $\onev -\psi$ in the eigenvalue equation $ \mathcal A v = \lambda v$, we obtain
\begin{align}
&s_1n + s_2\left(n - \sum_{i=1}^n \psi_i\right)= \lambda s_1 , \label{eq:57}
\\
&-s_1\left(n-\sum_{i=1}^n \phi_i\right) - s_2\left(n-\sum_{i=1}^n \psi_i - \sum_{i=1}^n \phi_i + \sum_{i=1}^n \psi_i \phi_i\right) =\lambda s_2 . \label{eq:58}
\end{align}
We rewrite \cref{eq:57,eq:58} as 
\begin{equation}
n\mathcal B \begin{bmatrix}
s_1 \\ s_2
\end{bmatrix} 
= \lambda 
\begin{bmatrix}
s_1 \\ s_2
\end{bmatrix} 
, 
\end{equation}
where we have used the notations in~\cref{eq:def:langphipsirnv,eq:def:langphipsirang}. Therefore, $\lambda$ is an eigenvalue of the matrix~$n\mathcal B$, as desired.

\section*{Acknowledgments} 

N.M.~acknowledges the support provided through JST CREST Grant Number JPMJCR1304 and the JST ERATO Grant Number JPMJER1201, Japan. V.M.P.~acknowledges the support of the US National Science Foundation under grant CAREER-ECCS-1651433.

\end{document}